%% file: OnJitter_Letter.tex
\DocumentMetadata{lang=en-US}

\newif\ifIEEE 
\IEEEtrue

\ifIEEE
    \documentclass{IEEEtran}
    \usepackage{algorithmic}
    \def\BibTeX{{\rm B\kern-.05em{\sc i\kern-.025em b}\kern-.08em
    T\kern-.1667em\lower.7ex\hbox{E}\kern-.125emX}}
\else
    \documentclass[a4paper,12pt,onecolumn,american]{article}
    \usepackage{fancyhdr}
    \usepackage[width=0.8\textwidth,font=footnotesize,format=hang,labelfont=bf]{caption}

\fi

\usepackage{hyperref}
\hypersetup{
    colorlinks=true,
    linkcolor=blue,
    citecolor=red,
    urlcolor=magenta,
    pdfstartview=FitH,
    pdfauthor={Dieter Schwarzmann, Simon Käser}
}

\usepackage{cite}
\usepackage{textcomp}
\usepackage{microtype} 

\usepackage{graphicx}
\usepackage{pgfplots,tikz}
\usetikzlibrary{arrows.meta, positioning, shapes.symbols}
\usetikzlibrary{calc}
\pgfplotsset{compat=1.18}
\usepackage{graphicx}

\usepackage{babel}

\usepackage{amssymb,amsmath,amsfonts,mathtools}
\mathtoolsset{showonlyrefs}
\usepackage{mathrsfs}

\usepackage{pdfpages}  
\usepackage{xspace}     

\usepackage{Dieter}

\newcommand{\ts}{t\sr{s}}
\newcommand{\tj}{\tilde{t}_k} 
\newcommand{\tk}{t_k}
\newcommand{\Dk}{\Delta \tk}
\newcommand{\dk}{_{\mathrm{d,}k}}
\newcommand{\etal}{et\,al.\xspace}
\begin{document}

\date{\today}
\title{On the Effect of Sampling-Time Jitter}
\author{Dieter Schwarzmann and Simon Käser
\thanks{Submitted for review in IEEE
Transactions on Control Systems Technology on September 4, 2025. With thanks to Prof. Jan Lunze and Prof. Andrea Iannelli for feedback and corrections .}
\thanks{Dieter Schwarzmann is an adjunct professor at IST of the University of Stuttgart, Germany and owner of www.systemwissenschaften.de, an engineering services company.
(e-mail: dieter.schwarzmann@ist.uni-stuttgart.de). }
\thanks{Simon Käser is a graduate student in Engineering Cybernetics at the University of Stuttgart, Germany
(e-mail: simonwilhelmkaeser@gmail.com).}}
\maketitle

\begin{abstract}
    This brief, aimed at practitioners, offers an analysis of the effect of sampling-time jitter, \ie the error produced by execution-time inaccuracies. We propose reinterpreting jitter-afflicted linear time-invariant systems through equivalent jitter-free analogs. By constructing a perceived system that absorbs the effects of timing perturbations into its dynamics, we find an affine scaling of jitter. We examine both measurement and implementation scenarios, demonstrating that the presence of jitter effectively scales the system matrices. Moreover, we observe that, in the Laplace domain, jitter can be interpreted as a frequency scaling. 
\end{abstract}

\section{Introduction}

Timing imperfections, known as \emph{sampling-time jitter} (also referred to as \emph{aperiodic/nonuniform sampling}), cause the actual sampling interval to deviate from its nominal value. That is, although an operation is nominally expected to occur every $\ts$, actual executions occur at intervals of $\ts + \tilde{t}$. We define jitter $\tilde t$ as an unknown, bounded, stochastic deviation from the nominal sampling time $\ts$.

Our work is motivated by the introduction of central computers in modern automotive architectures. There, function executions are triggered by timers of an operating system, which are subject to preemption due to increased resource contention, leading to variations in task execution timing.  

This brief aims to provide practitioners with a simple and interpretable analysis of how jitter affects system dynamics. 
While the theoretical foundations of sampled-data systems are well-established, we could not find clear, direct explanations on how to assess the specific effects of timing jitter on the perceived open-loop behavior of the affected linear time-invariant system. To this end, we find the corresponding system which at nominal sampling produces the identical I/O behavior on a sample-by-sample basis as the system under jitter. The results we present may appear obvious to some, but when faced with similar situations, we could not find a concise reference analyzing how jitter-sampled continuous-time systems can be interpreted. 

This paper is organized as follows.
In Section \ref{sec:approach} we present the problem description, our approach, and notation. Section~\ref{sec:effects} presents our results, deriving scaling relationships dependent on jitter for the $A$- and $B$-matrices as well as a Laplace-domain interpretation. Section~\ref{sec:examples} gives illustrative examples. 

\section{Literature Review}

Sampling‐time jitter originates from research on nonuniform sampling, where early studies characterized reconstruction errors and spectral distortion under random timing perturbations \cite{Balakrishnan1962,Brown1965,Marvasti2001}. In system identification, Eng and Gustafsson proposed estimators that compensate for stochastic sampling‐time errors \cite{EngGustafsson2006,EngGustafsson2008}, and Sei analyzed how jitter statistics bias discrete‐time parameter estimates \cite{Sei2013}. These efforts enhance discrete‐time estimation but do not reinterpret the underlying continuous‐time dynamics, helping the practitioner interpret behavior.

In control engineering, aperiodic sampling is commonly modeled as a time‐varying input delay. Mikheev \etal \cite{MikheevSobolevFridman1988}, Wittenmark \etal \cite{Wittenmark1995Timing}, and Lincoln \cite{Lincoln2002} proposed delay‐compensation schemes for digital controllers while robust LMI techniques ensure closed‐loop stability under bounded jitter \cite{Oishi2010}. Surveys by Hetel \etal \cite{Hetel2017}, and Zhang \etal \cite{Zhang2023survey} synthesize these methods. Unlike these delay‐based or closed‐loop analyses, our work focuses on deriving a jitter‐free continuous‐time analogue whose sample‐by‐sample I/O behavior matches that of the jitter‐afflicted system. To our knowledge, no prior study has offered this direct continuous‐time, non-conservative reinterpretation of sampling‐time jitter.

\section{Problem Description and Assumptions}\label{sec:approach}
\begin{figure} [ht]
  \centering
    \input{jitterfig}
  \caption{Sample and hold of input $u$ and output $y$ at times $t_k$}
  \label{fig:sampling}
\end{figure}
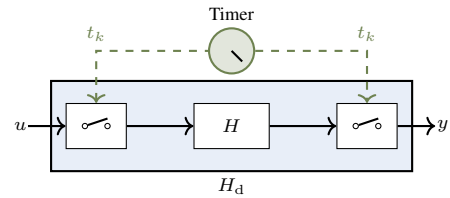
In \figref{fig:sampling}, a timer regularly triggers the execution of a function at points in time $\tk$, with $k \in \ZZ$ being the count of samples. We write a subscript $k$ to denote a sample-dependent variable. That function can be simply taking measurements of $u_k$ and $y_k$ at points in time $\tk$, \eg for offline system identification purposes, or it can be starting the execution of a computation of $y_k$ given $u_k$ (\eg controller or filter).  The execution instances occur every $\Dk = t_{k+1}-t_k$, designed to be the constant nominal sampling time $\Dk = \ts = \mathrm{const} \; \forall{k}$. We are concerned with interpreting the effect resulting from an additive, unknown, and time-varying jitter $\tj$ to that sampling time, varying the duration of every sample $k$. We denote systems and associated matrices with a tilde ($\tilde \cdot$) when under the influence of jitter and without for periodic sampling.

When asking for the effect that jitter produces, we distinguish between two different problem statements:
\begin{enumerate}
    \item[A)] \textbf{Plant Measurement}: Given a true continuous-time system $\tilde H$ which is sampled under jitter every $\ts+\tj$; which continuous-time system $H$ is perceived, assuming these samples had been taken at nominal sampling time $\ts$, \ie which system $H$ produces the identical I/O behavior but when discretized without jitter?
    \item[B)] \textbf{Controller or Filter Implementation}: Given a designed continuous-time system $H$ which was discretized with $\Dk = \ts = \mathrm{const}$ to obtain a discrete-time implementation $H\sr{d}$; what true behavior $\tilde H$ results when this discrete-time system $H\sr{d}$ is executed with jitter $\ts + \tj$?
\end{enumerate}
\begin{note} Both questions A) and B) are reversed to each other in the sense, that the given system in one is the desired solution in the other. 
Moreover, the given system is mostly considered time invariant and is then sampled with jitter - the perceived system or the effective resulting system will be time-variant, varying every sample $k$.
\end{note}

We solve both questions with the same approach by finding a relationship between $H$ and $\tilde H$ dependent on jitter. We leave it to the reader to pick their problem statement and solve accordingly. The examples at the end of this paper show both cases.

\paragraph*{Expectation for Case A}  
Assuming positive jitter, each sample sees a longer time step. Yet, since the nominal interval is assumed, the measurement appears contracted relative to the true movement. In other words, the perceived poles of the system are expected to be higher depending on the jitter.
\\[1mm]
\noindent For meaningful results, we restrict our investigation to practical choices for sampling times and jitter magnitudes.

\begin{ass}[Choice of nominal sampling time]\label{ass:ts_factor_alpha}
    In order to ensure that any oscillatory behaviors are captured without aliasing and can be recovered uniquely (\cf Lemma~\ref{lem:ts_for_no_aliasing}), we set a maximum sampling time depending on the highest imaginary part of all poles $\omega_{\max} = \max_i \abs{\Im{ \lambda_i}}, \forall i$  as
    \begin{equation}
        \ts \le \frac{\pi}{\omega_{\max}}.
    \end{equation}
\end{ass}

\begin{ass}[Jitter size]\label{ass:tj_factor_ts}
    Jitter can assume positive or negative values (allowing for premature sampling), but is limited to $\tj > -\ts$ (to avoid sampling past data or resampling the same instant). Moreover, we assume jitter is a time-varying fraction $\varepsilon_k$ of the nominal sampling time:
    \begin{equation}
        \tj = \varepsilon_k\, \ts,\quad \text{with } \varepsilon_k > -1.
    \end{equation}
    This definition enables practitioners to interpret jitter directly as a percentage of $\ts$ (for example, 10\% jitter corresponds to $\varepsilon_k = \pm 0.1$), which is a typical requirement for automotive systems.
\end{ass}
\begin{note}
    In practice, we would recommend limiting the upper bound of jitter to $\tj < \ts$ in order to avoid full sample misses (\ie $\varepsilon_k \in (-1,1)$) as these may be better treated as delays. However, all results hold for $-1 < \varepsilon_k < \infty$.
\end{note}

\section{Effect of Jitter} \label{sec:effects}
We will answer the question: Given a continuous-time system $\tilde H$, sampled with jitter $\Dk = \ts +\tj$; which system $H$ is perceived if the same I/O samples taken are assumed to have been recorded at uniform sampling $\Dk = \ts$. In other words, which two systems $H$ and $\tilde H$, produce the identical I/O behavior $(u_k, y_k)$, although at different sampling $\ts$ and $\ts +\tj$, respectively. This is achieved if they produce the identical discrete-time system $H\sr d$.

\begin{rem} The critical reader will have noticed that equivalence of I/O signals on a sample-by-sample basis but with different sampling times will mean a different temporal behavior. This is true, but also meaningless when it comes to implementing algorithms on computers. The computer or control unit sees samples and we are concerned with behavior that the computer sees and how to interpret that back in continuous time. Thus, a sample-by-sample equivalence is what is needed.
\end{rem}

\subsection{Preliminaries} 
The linear continuous-time system $H$ is
\begin{equation}\label{eq:G_definition}
\begin{split}
    H:\quad \dot{x}(t) &= A_k x(t) + B_k u(t),\quad x(0)=x_0, \\
    y(t) &= C x(t) + D u(t).
\end{split}
\end{equation}
The subscript $k$ in $A_k$ and $B_k$ indicates that both matrices are constant within one sample, but may change sample-to-sample.
Also note, that the system under jitter $\tilde H$ is described identically but with $\tilde A_k$ and $\tilde B_k$.

The exact discretization of equation \eqref{eq:G_definition} for a single-sample timestep $\Dk$ involves solving the differential equation at each discrete instance $k=1,2,\ldots$, and assuming a constant input over each interval (i.e., $u(t)=u_k$ for $t\in [\tk,t_{k+1})$). The discrete-time system is then described by matrices $A\dk$ and $B\dk$, both varying each sample $k$ by
\begin{equation}\label{eq:Gd_definition}
\begin{split}
   x_{k+1} &= A\dk\, x_k + B\dk\, u_k,\quad x[0]=x_0,\\[1mm]
    y_k &= C x_k + D u_k,
\end{split}
\end{equation}
with \cite{levine2010control} $\Dk = \ts$
\mathtoolsset{showonlyrefs=false}
\begin{subequations} \label{eq:AdBd_definition}
\begin{equation} \label{eq_test}
    A\dk = e^{A_k\,\ts}, \quad  B\dk = \int_0^{\ts} e^{A_k\,\tau}\,d\tau\, B_k,
\end{equation}
and equivalently for $\Dk = \ts+\tj$
\begin{equation}
A\dk = e^{\tilde A_k(\ts+\tj)},\quad  B\dk = \int_0^{\ts+\tj} e^{\tilde A_k\,\tau}\,d\tau\, \tilde B_k.
\end{equation}
\end{subequations}
\mathtoolsset{showonlyrefs=true}
\subsection{Impact on the \texorpdfstring{$A$-Matrix}{A-Matrix}}

We state the main result of this section as the following theorem.
\begin{thm}[Jitter scales the $A$-matrix]\label{thm:jitter_scales_Ahat}
    A system $\tilde H$ under jitter $\tj = \varepsilon_k \ts$ with Assumptions~\ref{ass:ts_factor_alpha} and \ref{ass:tj_factor_ts} will be perceived as a system $H$, sampled at $\ts$ where each sample $k$ of the $A$-matrix appears scaled by
    \begin{equation}
        A_k = \tilde{A}_k\cdot \frac{\ts+\tj}{\ts} = \tilde{A}_k\cdot (1+\varepsilon_k).
    \end{equation}
\end{thm}

\begin{proof}
    From equation \eqref{eq:AdBd_definition}, we have
    $A\dk = e^{\tilde{A}_k\,(\ts+\tj)}$ and $A\dk = e^{A_k\,\ts}$.
    Equating these and taking the matrix logarithm (which is unique on the principal branch under Assumption~\ref{ass:ts_factor_alpha}, see Lemma \ref{lem:ts_for_no_aliasing}) yields
        $\tilde{A}_k\cdot (\ts+\tj) = A_k\cdot \ts$.
\end{proof}



\subsection{Impact on the \texorpdfstring{$B$-Matrix}{B-Matrix}}
We now state the analogous result for the $B$-matrix:

\begin{thm}[Jitter scales the $B$-matrix]\label{thm:jitter_scales_Bhat}
    Jitter $\tj = \varepsilon_k\, \ts$ scales the resulting $B$-matrix as
    \begin{equation}
        B_k = \tilde{B}_k\cdot \frac{\ts+\tj}{\ts} = \tilde{B}_k\cdot (1+\varepsilon_k).
    \end{equation}
\end{thm}

\begin{proof}
Equating
$B\dk=\int_{0}^{t_{s}}e^{A\tau}d\tau B$ and $K=B\dk=\int_{0}^{t_{s}(1+\varepsilon)}e^{\tilde A\tau} d\tau \tilde B$ we obtain equal integral limits by the change of variable in $K$: $\tau=(1+\varepsilon)\sigma$, implying $d\tau=(1+\varepsilon)d\sigma$. Thus,
\begin{equation}
K=(1+\varepsilon)\int_{0}^{t_{s}}e^{\tilde A(1+\varepsilon)\sigma}\,d\sigma \tilde B
=(1+\varepsilon)\int_{0}^{t_{s}}e^{A\sigma}\,d\sigma \tilde B,
\end{equation}
with $A=(1+\varepsilon)\tilde A$ from Theorem \ref{thm:jitter_scales_Ahat}. 

We can now equate
\begin{equation}
\int_{0}^{t_{s}}e^{A\tau}\,d\tau B
=(1+\varepsilon)\int_{0}^{t_{s}}e^{A\tau} \,d\tau \tilde B.
\end{equation}
leading to the stated result.
\end{proof}
\begin{note}
    The integral $\int_{0}^{t_{s}}e^{A\tau}\,d\tau$ is uniquely invertible under Assumption \ref{ass:ts_factor_alpha}. Due to space limitation we will not show this.
\end{note}

\subsection{Impact of Jitter on the Perceived Transfer Function}
When jitter is constant, we can analyze its effect directly in the Laplace domain. The following theorem summarizes the result:

\begin{thm}[Jitter scales frequency]
    Assuming constant jitter $\tj = \varepsilon \ts = \mathrm{const}$, the transfer function of the perceived system $H(s)$ is given by
    \begin{equation}
        H(s) = C\left(\frac{s}{1+\varepsilon}I - A\right)^{-1}B + D.
    \end{equation}
    Equivalently, jitter scales the Laplace variable $s$ of the original system:
    \begin{equation}
        H(s) = \tilde{H}\left(\frac{s}{1+\varepsilon}\right).
    \end{equation}
\end{thm}

\begin{proof}
    Starting from
        $H(s) = C(sI-A)^{-1}B + D$
    and using Theorems~\ref{thm:jitter_scales_Ahat} and \ref{thm:jitter_scales_Bhat} with $A=(1+\varepsilon)\tilde A$ and $B=(1+\varepsilon)\tilde B$, we obtain
    \begin{equation} \label{eq:proof_Bscaling}
        H(s) = C\left(sI-(1+\varepsilon)\tilde A\right)^{-1}(1+\varepsilon) \tilde B + D.
    \end{equation}
    Notice that
    \begin{equation}
        sI - (1+\varepsilon) \tilde A = (1+\varepsilon)\left(\frac{s}{1+\varepsilon}I-\tilde A\right).
    \end{equation}
    Taking the inverse we have
    \begin{equation}
        \left(sI-(1+\varepsilon)\tilde A\right)^{-1} = \frac{1}{1+\varepsilon}\left(\frac{s}{1+\varepsilon}I-\tilde A\right)^{-1}.
    \end{equation}
    Inserting back into equation \eqref{eq:proof_Bscaling} gives the stated result by canceling the positive scalar $1+\varepsilon$.
\end{proof}

The scaling of the A-matrix by $(1 + \varepsilon)$ implies that the eigenvalues are scaled accordingly. For a stable system with eigenvalues having negative real parts, a positive $\varepsilon$ makes the real parts more negative, resulting in faster decay rates. Conversely, a negative $\varepsilon$ leads to slower decay. In the frequency domain, the transfer function $H(s) = \tilde H\left(\frac{s}{1 + \varepsilon}\right)$ indicates that the frequency response is stretched or compressed. For instance, resonant peaks appear at frequencies scaled by \((1 + \varepsilon)\), affecting the perceived bandwidth and resonance characteristics of the system.

\subsection{Summary}
Table \ref{tab:jitter_comparison} shows a consolidated comparison for the two use cases as a brief reference. It is valid on a sample-by-sample basis (i.e. the subscript $k$ is omitted for brevity) - with the exception of the transfer function, which is only valid for constant jitter.

\begin{table}[!h]
  \centering \scriptsize
  \caption{Nominal vs.\ jitter‐scaled models (Cases A and B)}
  \label{tab:jitter_comparison}
  \begin{tabular}{l|l|l}
    \hline
    Nominal ($\ts$) & Case A: Plant Measurement & Case B: Controller Implementation \\
    \hline
    $A$    & $A = (1 + \varepsilon) \tilde A$    & $\tilde A = A / (1 + \varepsilon)$   \\
    $B$    & $B = (1 + \varepsilon)\,\tilde B$                               & $\tilde B = B / (1 + \varepsilon)$                              \\
    $H(s)$ & $H(s) = \tilde H\!\bigl(\tfrac{s}{1+\varepsilon}\bigr)$ & $\tilde H(s) = H\!\bigl((1+\varepsilon)\,s\bigr)$    \\
    \hline
  \end{tabular}
\end{table}

\section{Examples} \label{sec:examples}
\subsection{Resulting LPV System}

If the sampling‐time error $\varepsilon_k$ can be measured online, the sampled plant admits a Linear Parameter‐Varying (LPV) realization whose scheduling variable is the jitter fraction.  In this formulation, the state‐space matrices depend affinely on $\varepsilon_k$, enabling standard gain‐scheduling or LPV synthesis.

\begin{equation}
  \dot x = A(\varepsilon_k) x + B(\varepsilon_k) u,
\end{equation}
where
\begin{equation}
  A(\varepsilon_k)=\left(1+\varepsilon_k\right)\,A,\quad
  B(\varepsilon_k)=\left(1+\varepsilon_k\right)\,B,
\end{equation}
and the scheduling parameter satisfies $\varepsilon_k\in(-1,1)$, ensuring each realization remains well‐posed. Note that $\varepsilon_k$ is a piecewise constant function, changing its value between samples.

With this LPV description, practitioners can apply established design tools (e.g.\ \cite{Shamma2008LPV}) to synthesize controllers that explicitly account for sampling‐time jitter.  Thus, jitter is handled directly by updating the plant model at each step according to the measured $\varepsilon_k$.

\subsection{Measuring a First-Order System}
Consider a first-order continuous-time system with transfer function
$ G(s) = \frac{1}{s + a}$, 
where \(a > 0\). The perceived transfer function under constant jitter \(\tj = \varepsilon \ts\) is
\[ \hat{G}(s) = G\left(\frac{s}{1 + \varepsilon}\right) = \frac{1 + \varepsilon}{s + a (1 + \varepsilon)}. \]
Thus, the pole of the perceived system is at \(-a (1 + \varepsilon)\), which is more negative if \(\varepsilon > 0\), indicating a faster perceived response. However, the DC gain remains \(\frac{1}{a}\), as
\[ \hat{G}(0) = \frac{1 + \varepsilon}{a (1 + \varepsilon)} = \frac{1}{a} = G(0). \]
This illustrates that while the dynamics are scaled, the steady-state behavior is preserved.

\subsection{PID Controller Under Constant Jitter}

For the controller‐implementation scenario (see Section~\ref{sec:approach}, case B), sampling‐time execution at $\ts(1+\varepsilon)$ rescales the continuous‐time state matrices by $(1+\varepsilon)^{-1}$ and, in the Laplace domain, corresponds to substituting 
\[
s\;\mapsto\;s(1+\varepsilon)
\]
in any designed controller $C(s)$.  Consequently, controller gains, and time constants are altered under jitter.

For the practitioner, it is worthwhile investigating the most prevalent control structure - even for constant jitter - in order to obtain some insights. Consider the realizable PID law $C(s)$ 
\begin{equation}
C(s)=K_{p}+\frac{K_{i}}{s}+K_{d}\frac{s}{\tau_{d}s+1}.
\end{equation}
Under constant jitter $\varepsilon$, the effective controller becomes
\begin{equation}
\tilde C(s)=
K_{p}
+\frac{K_{i}}{(1+\varepsilon)s}
+K_{d}(1+\varepsilon)\frac{s}{\tau_{d}(1+\varepsilon)s+1}.
\end{equation}

Assuming positive jitter $\varepsilon>0$, the controller parameters modify as follows:
\begin{itemize}
\item Proportional gain $K_{p}$ remains unchanged.
\item Integral gain reduces to $K_{i}/(1+\varepsilon)$, slowing disturbance rejection.
\item Derivative gain increases to $K_{d}(1+\varepsilon)$, boosting mid‐band response.
\item Derivative time constant expands to $\tau_{d}(1+\varepsilon)$, shifting the roll‐off toward lower frequencies.
\end{itemize}

For a representative jitter of $\varepsilon=0.1$, the integral and derivative gains change by approximately 10\%, which underpins the rule‐of‐thumb that up to 10\% sampling‐time jitter is acceptable. If $\varepsilon<0$, these effects reverse, accelerating integral action and reducing derivative action.

\section{Conclusion}
We have modeled sampling‐time jitter as a relative perturbation $\varepsilon$ of the nominal sampling period and shown that it induces a simple scaling law in both measurement and implementation scenarios. Every sample’s continuous‐time $A$‐ and $B$‐matrices are multiplied by $(1+\varepsilon)$ when recovering a jitter‐free model, or divided by $(1+\varepsilon)$ when executing a nominal design under jitter. In the Laplace domain, this corresponds to a frequency scaling of the transfer function.
Although elementary, this brief collects these facts into a single, concise reference for continuous‐time interpretation.

Our results use a sample-by-sample deterministic affine scaling of the $A$‐ and $B$‐matrices. In practice, the jitter fraction $\varepsilon_k$ is random. This approach is open for future extensions exploiting the affine relationship (jitter can be regarded as a multiplicative disturbance) with jitter as a random variable with a known worst-case probability-density function.
Adding time-varying delays into the jittered system behavior would lead the results presented here to distributed or networked control, where jitter arises from communication delays and asynchronous execution. 

Based on these results, our future work will focus on stochastic analysis and robust control design approaches.
\appendix

\section{Technical Lemmas}

The following lemma shows that transformation from continuous time to discrete time and back is unique if the sampling time is chosen small enough.

\begin{lem}\label{lem:ts_for_no_aliasing}
    A continuous-time pole $\lambda = \sigma+\im \omega$, mapped to a discrete-time pole $\lambda\sr{d} = e^{\lambda\,\ts}$, can be uniquely recovered via the logarithm
        $\lambda = \frac{1}{\ts}\ln\!\left(\lambda\sr{d}\right)$
    if
    \begin{equation}
        \ts \le \frac{\pi}{|\omega|}.
    \end{equation}
\end{lem}

\begin{proof}
    Writing the discrete-time pole in polar form,
        $\lambda\sr{d} = e^{\sigma\,\ts}\,e^{\im\omega\,\ts},$
    the argument of that complex number is
        $\arg\!\left(\lambda\sr{d}\right) = \omega\,\ts$.
    To recover the continuous-time pole uniquely, the argument must lie within the principal branch,
        $-\pi < \omega\,\ts \le \pi$.
    This condition implies
        $|\omega\,\ts|\le \pi$ 
   leading to the stated result.
\end{proof}

\begin{note}
    This criterion is related to the Nyquist-Shannon sampling theorem. For a pole oscillating at frequency $|\omega|$, 
    converting the frequency of $\omega$ from rad/s to Hz reveals the known rule $f\sr s \geq 2 \abs{\omega} (\text{Hz}) $.
\end{note}

\bibliographystyle{IEEEtran}
\bibliography{Literature}

\end{document}

%% file: jitterfig.tex
\begin{tikzpicture}[
  block/.style = {draw, minimum width=1.0cm, minimum height=0.6cm},
  swblock/.style = {draw, minimum width=0.8cm, minimum height=0.6cm},
  signal/.style = {->, thick},
  dashedline/.style = {thick, dashed},
  circ/.style = {circle, draw, minimum size=0.6cm},
  every node/.style={font=\scriptsize}
  ]

\definecolor{myblue}{RGB}{210,220,240}
\definecolor{mygreen}{RGB}{110,130,80}

\filldraw[fill=myblue!50, draw=black, thick] (-2.4,-0.6) rectangle (2.4,0.6);

\node at (0,-0.8) {$H\sr{d}$};

\node[swblock, fill=white] (sw1) at (-1.8,0) {};
\node[block, fill=white] (G) at (0,0) {$H$};
\node[swblock, fill=white] (sw2) at (1.8,0) {};

\node at ($(sw1.west) + (-0.6,0)$) {$u$};
\node at ($(sw2.east) + (0.6,0)$) {$y$};

\draw[signal] ($(sw1.west) + (-0.5,0)$) -- (sw1.west);
\draw[signal] (sw1.east) -- ++(0.3,0) -- (G.west);
\draw[signal] (G.east) -- ++(0.3,0) -- (sw2.west);
\draw[signal] (sw2.east) -- ++(0.5,0);

\coordinate (sw1center) at (sw1.center);
\node[circle, draw=black, minimum size=2pt, inner sep=0pt] (sw1dotL) at ($(sw1center) + (-0.15,0)$) {};
\node[circle, draw=black, minimum size=2pt, inner sep=0pt] (sw1dotR) at ($(sw1center) + (0.15,0)$) {};
\draw[thick] (sw1dotL) -- ($(sw1dotR) + (0,0.1)$);

\coordinate (sw2center) at (sw2.center);
\node[circle, draw=black, minimum size=2pt, inner sep=0pt] (sw2dotL) at ($(sw2center) + (-0.15,0)$) {};
\node[circle, draw=black, minimum size=2pt, inner sep=0pt] (sw2dotR) at ($(sw2center) + (0.15,0)$) {};
\draw[thick] (sw2dotL) -- ($(sw2dotR) + (0,0.1)$);

\node[circ, draw=mygreen, thick, fill=mygreen!20] (clock) at (0,1.0) {};
\draw[thick] (clock.center) -- ++(4pt,-4pt);

\node at (0,1.5) {\text{Timer}};

\draw[dashedline, mygreen, <-] 
  (sw1) -- 
  ($(sw1|-clock)$) node[above, text=mygreen] {$t_k$} -- 
  (clock);

\draw[dashedline, mygreen, ->] 
  (clock) -- 
  ($(sw2|-clock)$) node[above, text=mygreen] {$t_k$} -- 
  (sw2);

\end{tikzpicture}